\documentclass[review]{elsarticle}
\makeatletter
\def\ps@pprintTitle{%
 \let\@oddhead\@empty
 \let\@evenhead\@empty
 \def\@oddfoot{\centerline{\thepage}}%
 \let\@evenfoot\@oddfoot}
\makeatother

\usepackage{lineno,hyperref}
\modulolinenumbers[5]

\usepackage{tabularx}
\usepackage{booktabs}

\usepackage{adjustbox} 

\usepackage{amssymb} 
\usepackage{stix}     
\usepackage{amsmath} 
\usepackage{amsthm}
\usepackage{cleveref}
\usepackage{algpseudocode,algorithm,algorithmicx}

\usepackage[english]{babel}
\newtheorem{theorem}{Theorem}[section]

\newtheorem{lemma}[theorem]{Lemma}
\newtheorem{definition}[theorem]{Definition}

\usepackage{csquotes} 

\graphicspath{{figures/}}
\usepackage{caption}
\usepackage{subcaption}

\usepackage{xcolor}
\newif \ifhidecomment
\hidecommenttrue
\hidecommentfalse 

\usepackage{comment}

\journal{Discrete Mathematics}

\begin{document}

\begin{frontmatter}

\title{Finding $(s,d)$-Hypernetworks in F-Hypergraphs is NP-Hard}


\author[1]{Reynaldo Gil Pons}
\author[2,4]{Max Ward}
\author[3]{Loïc Miller}


\address[1]{University of Luxembourg, 6 Avenue de la Fonte, L-4364 Esch-sur-Alzette, Luxembourg}
\address[2]{University of Adelaide, Waite Road, SA 5005 Adelaide, Australia}
\address[3]{University of Strasbourg, 4 Rue Blaise Pascal 67081 Strasbourg,
France}
\address[4]{Harvard University, 52 Oxford Street, Cambridge, MA 02138, USA}

\begin{abstract}

We consider the problem of computing an $(s,d)$-hypernetwork in an acyclic F-hypergraph.
This is a fundamental computational problem arising in directed hypergraphs, and is a foundational step in tackling problems of reachability and redundancy.
This problem was previously explored in the context of general directed hypergraphs (containing cycles), where it is NP-hard, and acyclic B-hypergraphs, where a linear time algorithm can be achieved.
In a surprising contrast, we find that for acyclic F-hypergraphs the problem is NP-hard, which also implies the problem is hard in BF-hypergraphs.
This is a striking complexity boundary given that F-hypergraphs and B-hypergraphs would at first seem to be symmetrical to one another.
We provide the proof of complexity and explain why there is a fundamental asymmetry between the two classes of directed hypergraphs.
\end{abstract}


\end{frontmatter}


\section{Introduction}\label{sec:introduction}

Directed hypergraphs are widely studied and have many applications~\cite{gallo1993directed,ausiello2017directed,ausielloyz1992optimal,ausiello2001directed,thakur2009linear,cambini1997flows}. 
Our focus is computing $(s,d)$-hypernetworks in directed hypergraphs, which are fundamental for answering questions about reachability and redundancy~\cite{pretolani2013finding,volpentesta2008hypernetworks}. We begin by providing necessary definitions related to directed hypergraphs.

\subsection{Hypergraphs}\label{subsec:hypergraphs}

Hypergraphs are a generalization of classical graphs where sets of elements are connected by a single hyperedge.
In this paper, we are concerned with directed hypergraphs and thus refer to directed hypergraphs simply as hypergraphs.

Each edge in a hypergraph is a directed set-to-set mapping, where the source set is called the \textit{tail} of the edge and the target set is called the \textit{head} of the edge.
A hypergraph can be defined as follows:
\begin{definition}[Hypergraph]
A directed hypergraph is a pair $\mathcal{H} = (V, E)$, where $V = \{v_1, v_2, \ldots, v_n\}$ is the set of hypervertices, and $E = \{e_1, e_2, \ldots, e_m\}$ is the set of hyperedges.
\end{definition}

Note that we use vertex and hypervertex when describing a hypergraph interchangeably. Likewise, edge and hyperedge are used interchangeably.

\begin{definition}[Hyperedge]
A directed hyperedge is an ordered pair $e = (T_e, H_e)$, where $T_e \subseteq V$ is the tail of $e$ and $H_e \subseteq V \setminus T_e$ is its head.
$T_e$ and $H_e$ may be empty.
\end{definition}

\begin{figure}
     \centering
     \begin{subfigure}[b]{0.45\textwidth}
         \centering
         \includegraphics[width=0.7\textwidth]{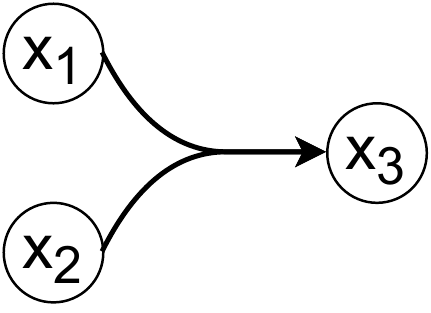}
         \caption{A backward hyperedge (B-edge).}
         \label{fig:b-edge}
     \end{subfigure}%
     \hfill
     \begin{subfigure}[b]{0.45\textwidth}
         \centering
         \includegraphics[width=0.7\textwidth]{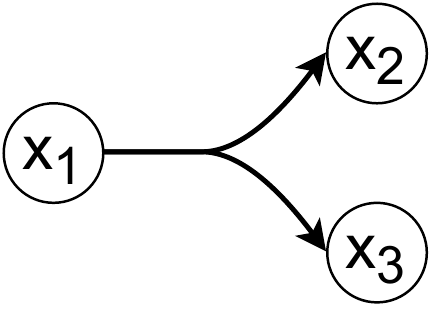}
         \caption{A forward hyperedge (F-edge).}
         \label{fig:f-edge}
     \end{subfigure}
     \hfill
        \caption{Types of hyperedges in a hypergraph}
        \label{fig:hg-edges}
\end{figure}

Gallo \textit{et al.}~\cite{gallo1993directed} defines two different types of hyperedge.
A \textit{backward} hyperedge (B-edge) is a hyperedge $e = (T_e, H_e)$, with $|H_e| = 1$.
A \textit{forward} hyperedge (F-edge) is a hyperedge $e = (T_e, H_e)$, with $|T_e| = 1$.
Figure~\ref{fig:hg-edges} represents both a B-edge (Fig.~\ref{fig:b-edge}) and an F-edge (Fig.~\ref{fig:f-edge}).

It follows that a B-hypergraph is a hypergraph whose hyperedges are B-edges.
Similarly, an F-hypergraph is a hypergraph whose hyperedges are F-edges.
A BF-hypergraph is a hypergraph whose hyperedges are either B-edges or F-edges.

\begin{definition}[Path]
A path $P_{s,d}$ from $s$ to $d$ in a hypergraph $\mathcal{H}$ is a sequence $P_{st} = (v_1 = s, e_1, v_2, e_2, \ldots, e_q, v_{q+1} = d)$, where $s \in T_{e_1}$, $d \in H_{e_q}$ and $v_i \in \{H_{e_{i-1}}\} \cap T_{e_i}, i = 2, \ldots, q$.
\end{definition}

\begin{definition}[Acyclic Hypergraph]
A hypergraph $\mathcal{H}$ is \textit{acyclic} is there exists no path $P_{s,d}$ such that $s=d$.
\end{definition}

Several (almost) equivalent definitions of directed hyperpaths have been published~\cite{gallo1993directed,ausiello1998hypergraph,ausiello2005partially, thakur2009linear}.
The definition we give here is based on that of Ausiello \textit{et al.}~\cite{ausiello2005partially}.
Since we are interested in directed hyperpaths on directed hypergraphs, when we say hyperpath, it should be assumed that it is a directed hyperpath.

\begin{definition}[Subhypergraph]
A hypergraph $\mathcal{H}'=(V',E')$ is a subhypergraph of a hypergraph $\mathcal{H}=(V,E)$ if $V' \subseteq V$, $E' \subseteq E$ and for any edge $e \in E'$ it holds that $H(e) \subseteq V'$ and $T(e) \subseteq V'$.
\end{definition}

\begin{definition}[Hyperpath]
Let $\mathcal{H} = (V,E)$ be a directed hypergraph and let $s,d \in V$.
A hyperpath from $s$ to $d$ in $\mathcal{H}$ is a minimal subhypergraph (where minimality is with respect to deletion of vertices and edges) $\mathcal{H'} \subseteq \mathcal{H}$.
Further, the hyperedges comprising $\mathcal{H'}$ can be ordered in a sequence $\langle e_1, e_2, \dots, e_k \rangle$ such that for every edge $e_i \in \mathcal{H'}$ it is the case that $T(e_i) \subseteq \{s\} \cup H(e_1) \cup H(e_2) \cup \dots \cup H(e_{i-1})$ and $d \in H(e_k)$.
\label{def:hyperpath}
\end{definition}

\subsection{Hypernetworks}\label{sec:hypernetworks}

With the hypergraph definitions given in the preceding section, we can describe hypernetworks.

Hypernetworks were introduced by Volpentesta~\cite{volpentesta2008hypernetworks} who presents two types of hypernetwork, the $s$-hypernetwork and the $(s,d)$-hypernetwork.
Informally, the $s$-hypernetwork is the set of all elements in hyperpaths from a node $s$ to any node, and the $(s,d)$-hypernetwork is the set of all elements in hyperpaths from a node $s$ to a node $d$.

\begin{definition}[$(s,d)$-hypernetwork]
Consider a hypergraph $\mathcal{H}=(V,E)$.
Let $s, d \in V$ and let $\Pi_{s,d}$ be the set of hyperpaths from $s$ to $d$ in $\mathcal{H}$.

The $(s,d)$-hypernetwork in $\mathcal{H}$ is defined as the subhypergraph $\mathcal{H}_{s,d} = (V', E')$, where $E' = \underset{(\mathcal{V}, \mathcal{E}) \in \Pi_{(s,d)}}{\bigcup} \mathcal{E}$ and $V' = \underset{(\mathcal{V}, \mathcal{E}) \in \Pi_{(s,d)}}{\bigcup} \mathcal{V}$.
\label{def:s-d-hypernetwork}
\end{definition}

\begin{definition}[$s$-hypernetwork]
The $s$-hypernetwork in $\mathcal{H}=(V,E)$ is $\mathcal{H}_{s} = \underset{x \in V}{\bigcup} \mathcal{H}_{s,x}$.
\label{def:s-hypernetwork}
\end{definition}

Volpentesta describes polynomial algorithms for finding $s$-hypernetworks and for finding $(s,d)$-hypernetworks in acyclic B-hypergraphs. 
Note that computing an $(s,d)$-hypernetwork is NP-hard if cycles are permitted~\cite{volpentesta2008hypernetworks}.
Later, Pretolani \cite{pretolani2013finding} clarified the findings and terminology of Volpentesta and provided a linear time solution to finding $(s,d)$-hypernetworks in acyclic B-hypergraphs.

We explore the complexity of finding $(s,d)$-hypernetworks in acyclic F-hypergraphs.
It may surprise the reader to see that we prove the problem becomes NP-hard, despite that fact that F-hypergraphs are symmetric to B-hypergraphs.

\section{New Acyclic F-Hypergraph Results}

We show that computing an $(s,d)$-hypernetwork is NP-hard by reduction from the problem of forcing a particular edge in a hyperpath. The forced hyperpath edge problem is a new problem we introduce and prove is NP-complete.

\begin{definition}[$(s,d)$-Hypernetwork Problem]
The \textsc{$(s,d)$-Hypernetwork Problem} (SDHP) is to find the $(s,d)$-hypernetwork $\mathcal{H}_{s,d}$ given a hypergraph $\mathcal{H}$.
\end{definition}

\begin{definition}[Forced Hyperpath Edge Problem]
The \textsc{Forced Hyperpath Edge Problem} (FHEP) is a decision problem in which we decide if there exists any hyperpath $\mathcal{H}'=(V',E')$ in a hypergraph $\mathcal{H}$ between $s$ and $d$ that must contain an edge $e \in E'$.
\end{definition}

\begin{theorem}
\label{thm:fhep-and-sdhp}
If the FHEP is NP-complete, then the SDHP is NP-hard.
\end{theorem}
\begin{proof}
We can reduce any FHEP instance to a SDHP instance in polynomial time. The FHEP is defined between a source $s$ and destination $d$ on a hypergraph $\mathcal{H}$. Suppose we could solve the SDHP between $s$ and $d$ on $\mathcal{H}$. Any edge in the $(s,d)$-hypernetwork is forcible in the FHEP sense because there must be a hyperpath using it. Conversely, any edge not in the $(s,d)$-hypernetwork cannot be forcible in the FHEP sense because there must be no hyperpath using it. Thus, the FHEP is reducible to the SDHP.

If the FHEP is NP-complete, the SDHP must be NP-hard via our reduction.
\end{proof}

Next, we prove the FHEP is NP-complete.
\Cref{thm:fhep-and-sdhp} implies that the SDHP is therefore NP-hard.

\subsection{The Forced Hypergraph Edge Problem is NP-complete}\label{sec:fhep-np-complete}
Our proof involves a reduction from 3-SAT, which is the Boolean satisfiability problem restricted to exactly 3 variables per clause.
It is widely known to be NP-complete~\cite{cook1971complexity}.
The reduction is achieved by taking an instance of the 3-SAT problem and constructing a corresponding acyclic F-hypergraph such that any found hyperpath with a forced edge would imply a solution to the 3-SAT problem.

Assume a 3-SAT instance with variables $v_1, \ldots, v_n$ and clauses $c_1, \ldots, c_m$.
Our corresponding acyclic F-hypergraph construction contains a vertex and a pair of edges for each variable, and three vertices and three edges for each clause.

We start with an initial vertex $p_0$ which is the source of our hyperpath.
We also create a vertex $q_0$ that will be used to force the single must-use edge of our hyperpath, and a node $f$ which is the target of the hyperpath we need to find.
For each variable $v_i$, we create a node $p_i$.
Then for each clause $c_i$, we create three nodes $q_{i,1}, q_{i,2}, q_{i,3}$ which correspond respectively to the three literals of each clause.

A variable $v_i$ appears in its positive form in a set of $a$ clauses $x_1, x_2, \dots x_a$.
For each such clause, $v_i$ appears as either the first, second, or third literal, denoted by $y_1 \in \{1,2,3 \}, y_2 \in \{1,2,3 \}, \dots, y_a \in \{1,2,3 \}$.
A variable $v_i$ appears in negated form ($\lnot v_i$) in $b$ clauses. Call the corresponding clauses and literals $x'_1, x'_2, \dots x'_b$ and $y'_1 \in \{1,2,3 \}, y'_2 \in \{1,2,3 \}, \dots, y'_b \in \{1,2,3 \}$ respectively.

In our construction, for each variable $v_i$ we have a vertex $p_i$ and two hypergraph edges. The first edge is of the form $(\{p_{i-1} \},\{ p_i, q_{x_1, y_1}, q_{x_2, y_2}, \dots, q_{x_a, y_a} \})$, and it corresponds to assigning variable $v_i$ to false, thus blocking the clauses in the edge. The second edge is of the form $(\{p_{i-1} \},\{ p_i, q_{x'_1, y'_1}, q_{x'_2, y'_2}, \dots, q_{x'_b, y'_b} \})$, and it corresponds to assigning variable $v_i$ to true.

For each clause $c_i$ we construct three vertices $q_{i,1}, q_{i,2}, q_{i,3}$ and three edges. Each of the three edges corresponds to one of the three vertices. The $j$-th such edge has the form $(\{q_{i,j}\}, \{q_{i+1,1},q_{i+1,2},q_{i+1,3} \})$. Since $c_m$ is the last clause we construct it differently. We have $(\{q_{m,j}\}, \{ f \})$ for $j \in \{1, 2, 3 \}$.

Now, we must connect the $p$ vertices to the $q$ vertices. This is done by adding edges $(\{p_n\}, \{q_0\})$ and $(\{q_0\}, \{q_{1,1}, q_{1,2}, q_{1,3}\})$. Note that $(\{p_n\}, \{q_0\})$ is particularly important since it is the forced edge. We aim to find a hyperpath from $p_0$ to $f$ that must include $(\{p_n\}, \{q_0\})$.

\begin{figure}
\centering
\includegraphics[width=\textwidth]{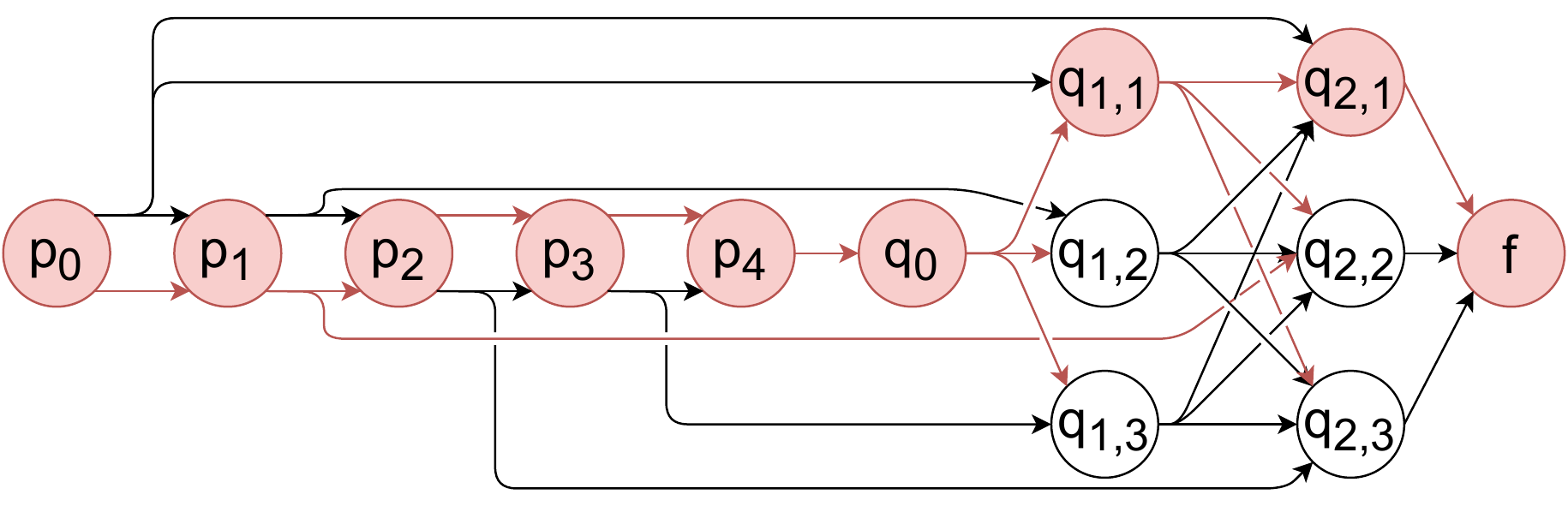}
\caption{An example of our acyclic F-hypergraph construction, corresponding to the following 3-SAT instance: $(v_1 \lor v_2 \lor \lnot v_4) \land (v_1 \lor \lnot v_2 \lor \lnot v_3)$. Red hyperedges and nodes indicate a valid hyperpath with the forced edge $(\{p_4\}, \{q_0\})$.}
\label{fig:3-sat-construction}
\end{figure}

To illustrate our construction, we consider the following 3-SAT instance: $(v_1 \lor v_2 \lor \lnot v_4) \land (v_1 \lor \lnot v_2 \lor \lnot v_3)$.
The F-hypergraph construction stemming from this instance can be found in Figure~\ref{fig:3-sat-construction}.
This instance is composed of four variables, $v_1, v_2, v_3, v_4$ and two clauses, which correspond respectively to vertices $p_1, p_2, p_3, p_4$ for the variables, vertices $q_{1,1}, q_{1,2}, q_{1,3}$ for clause 1, and vertices $q_{2,1}, q_{2,2}, q_{2,3}$ for clause 2.

Let us call our hypergraph construction $C(\phi)$ for a 3-SAT instance $\phi$.
To prove that finding a hyperpath in $C(\phi)$ with the forced edge $(\{p_n\}, \{q_0\})$ is equivalent to a valid variable assignment in $\phi$ we must first prove some lemmas.

\begin{lemma}
\label{thm:sat_acyclic_f_hypergraph}
$C(\phi)$ is an acyclic F-hypergraph.
\end{lemma}
\begin{proof}
Each edge in $C(\phi)$ has a tail containing a single vertex and is therefore an F-edge.
It follows from the definition of F-hypergraph that $C(\phi)$ is an F-hypergraph as it contains only F-edges.

Each edge in $C(\phi)$ progresses from the previous layer to the next. Therefore it is acyclic.
\end{proof}

\begin{lemma}
A hyperpath $\mathcal{H}'$ from $s$ to $d$ in an F-hypergraph $\mathcal{H}$ contains exactly one path $P_{s,d}$.
\label{thm:sat_proof_one_path}
\end{lemma}
\begin{proof}
We say that a path is contained in a hyperpath if the path is comprised of only vertices and edges in the hyperpath.

There is only one edge $e$ in $\mathcal{H}'$ such that $d \in H(e)$, because a hyperpath is minimal. Since $\mathcal{H}$ is an F-hypergraph, the tail of $e$ contains a single vertex $T(e)=\{v\}$.
We can apply the same reasoning to $v$ recursively.
That is, remove $d$ and $e$ from $\mathcal{H}'$ and treat $v$ as the target, then repeat.
Eventually, we will remove all the vertices from $\mathcal{H}'$ except $s$.
The vertices and edges we removed must be the only path from $s$ to $d$.
\end{proof}

\begin{lemma}
\label{thm:sat_one_edge_per_thing}
A hyperpath $\mathcal{H}$ in $C(\phi)$ that uses the edge $(\{p_n\}, \{q_0\})$ will contain contain exactly one edge $e_{v_i} \in \mathcal{H}$ for each variable in $v_i \in \phi$ and exactly one edge $e_{c_i} \in \mathcal{H}$ for each clause $c_i \in \phi$.
\end{lemma}
\begin{proof}
\Cref{thm:sat_proof_one_path} implies that $\mathcal{H}$ contains a single path from $p_0$ to $f$. Since we force the edge $(\{p_n\}, \{q_0\})$, there must also be a single path from $p_0$ to $q_0$.
The only way to get to $q_0$ is through $p_0 \dots p_n$.
For any $p_i$ where $i<n$ there are two edges to choose from, and we must pick exactly one of them on the path to $q_0$ because hyperpaths are minimal with respect to edge deletion.
By construction, $p_i$ corresponds to $v_i$, and we must pick a single edge $e_{v_i}$ from $p_{i-1}$ to $p_i$. Therefore, each $v_i$ has a single edge in $\mathcal{H}$.

Similarly, the only way to get from $q_0$ to the target, $f$, is through the $q$ vertices.
Since $\mathcal{H}$ contains a single path from $p_0$ to $f$, and since we force the edge $(\{p_n\}, \{q_0\})$, there must be a single path from $q_0$ to $f$.
The three vertices $q_{i,1}, q_{i,2}, q_{i,3}$ only have edges to $q_{i+1}$ or $f$.
Therefore, we must pass through exactly one $q_i$ vertex for each $i= 1 \dots m$.
By construction, the three $q_i$ vertices correspond to the clause $c_i$ and we must pick a single edge $e_{c_i}$ from $q_{i-1}$ to $q_i$.
Therefore, each $c_i$ has a single edge in $\mathcal{H}$.
\end{proof}

\begin{lemma}
\label{thm:sat_assignments}
A hyperpath $\mathcal{H}$ in $C(\phi)$ corresponds to a variable assignment that satisfies $\phi$.
\end{lemma}
\begin{proof}
\Cref{thm:sat_one_edge_per_thing} says we pick a single edge $e_{v_i}$ for each variable $v_i$.
By construction, $C(\phi)$ contains two edges for each $v_i$.
Exactly one of these edges corresponds to assigning $v_i$ to true, and is connected to all the clauses where $v_i$ appears in negated $\lnot v_i$ form.
The other of these edges corresponds to assigning $v_i$ to false, and is connected to all the clauses where $v_i$ appears in positive form.

If we pick the edge that corresponds to assigning $v_i$ to true, then all the clauses in which $\lnot v_i$ appears cannot be satisfied by our assignment to $v_i$.
They must be satisfied by at least one of the other two literals in the clause. Similarly, if we assign $v_i$ to false, it cannot satisfy clauses where $v_i$ appears in positive form.

Each edge $e_{c_i} \in \mathcal{H}$ is the single edge we pick for a clause $c_i$ (\Cref{thm:sat_one_edge_per_thing}).
This edge corresponds to deciding which of the three literals satisfies $c_i$.
If there is some literal that satisfies a clause under a variable assignment, then there is some valid choice for $e_{c_i}$.

If there existed a valid hyperpath solution in $C(\phi)$, then we could assign the variables in $\phi$ the corresponding values to $e_{v_i}$ for each $v_i$.
Then, each clause $c_i$ would be satisfied by the literal $e_{c_i}$.
This would be a variable assignment that satisfies $\phi$.
\end{proof}

\begin{lemma}
\label{thm:if-and-only-if}
A hyperpath $\mathcal{H}$ in $C(\phi)$ that corresponds to a variable assignment that satisfies $\phi$ exists \textbf{if and only if} there is a solution to $\phi$.
\end{lemma}
\begin{proof}
\Cref{thm:sat_assignments} shows that if we have a solution to $C(\phi)$, then we have a solution to $\phi$. We now show that if there is a solution to $\phi$, then there is a solution to $C(\phi)$.

Suppose there is a variable assignment satisfying $\phi$ that has no valid hyperpath solution in $C(\phi)$.
Since variables are assigned a value, it still holds that we have exactly one edge $e_{v_i} \in \mathcal{H}$ for each variable $v_i \in \phi$, that corresponds to a true or false assignment of the variable.
A variable assignment satisfying $\phi$ that has no valid hyperpath solution in $C(\phi)$ then means that for at least one clause $c_i$, there is no valid choice for $e_{c_i}$ as the literal which satisfies the clause.

However, if the variable assignment satisfies $\phi$, then there must be at least one literal per clause which we can use to satisfy the clause.
For each clause, we can therefore pick the edge that corresponds to the literal we use to satisfy the clause to create a valid hyperpath solution.
This contradicts our initial supposition.
\end{proof}

\begin{theorem}
\label{thm:sat_f_hypergraph_fhep}
The FHEP in acyclic F-hypergraphs is NP-complete.
\end{theorem}
\begin{proof}
Our construction $C(\phi)$ is an acyclic F-hypergraph (\Cref{thm:sat_acyclic_f_hypergraph}).

\Cref{thm:if-and-only-if} implies a hyperpath $\mathcal{H}$ in $C(\phi)$ from $p_0$ to $f$ that must use $(\{p_n\}, \{q_0\})$ exists if and only if a solution to 3-SAT instance $\phi$ exists. We can construct $C(\phi)$ for any 3-SAT instance $\phi$ in polynomial time, thus 3-SAT is polynomial time reducible to FHEP on an acyclic F-hypergraph.

The 3-SAT problem is NP-complete. Therefore, the FHEP in acyclic F-hypergraphs is NP-complete via our reduction.
\end{proof}

\Cref{thm:sat_f_hypergraph_fhep} and \Cref{thm:fhep-and-sdhp} together imply that the SDHP in acyclic F-hypergraphs is NP-hard. Note that the problem is already known to be NP-hard when cycles are permitted \cite{volpentesta2008hypernetworks}. Thus, SDHP is NP-hard in F-hypergraphs in general.

\section{Final Remarks}\label{sec:final-remarks}

We considered the problem of computing an $(s,d)$-hypernetwork in F-hypergraphs. A linear time solution exists for acyclic B-hypergraphs, and the problem is known to be NP-hard in general when cycles are permitted \cite{pretolani2013finding,volpentesta2008hypernetworks}. We use a reduction from 3-SAT to the forced hypergraph edge problem and thence to the problem of finding an $(s,d)$-hypernetwork to prove the problem NP-hard for acyclic F-hypergraphs.

There are two final remarks to make related to our new hardness proof. First, it implies the problem is also hard for acyclic BF-hypergraphs since they include acyclic F-hypergraphs. This is a new result. Second, one may be surprised that B-hypergraphs and F-hypergraphs have different complexity results. It seems intuitive that they are symmetrical: the directions of each edge could be reflected and $s$ and $d$ could be swapped. However, there is one subtle asymmetry that breaks this construction.

Hyperpaths are minimal with respect to edge deletion. So, in a hyperpath from $s$ to $d$, there is only a single edge $e$ such that $d \in H(e)$. However, they may be more than one edge $s \in T(e')$. This asymmetry means that finding a hyperpath in a reversed F-hypergraph may not be valid. Further, this property is sufficient to make the problem tractable in B-hypergraphs. Future research may find that this property is powerful enough to solve other algorithmic problems on B-hypergraphs.

\bibliography{main}

\end{document}